\documentclass[a4paper]{cas-dc}

\usepackage[numbers,sort&compress]{natbib}
\usepackage{amsthm}
\usepackage{amsmath,amssymb,amsfonts}

\usepackage{algorithmic}
\usepackage{graphicx}
\usepackage{textcomp}
\usepackage{xcolor}
\usepackage{caption}
\usepackage{lineno} 
\usepackage{float}
\usepackage{subcaption}

\newtheorem{theorem}{Theorem}
\newtheorem{remark}{Remark}

\def\tsc#1{\csdef{#1}{\textsc{\lowercase{#1}}\xspace}}
\tsc{WGM}
\tsc{QE}
\tsc{EP}
\tsc{PMS}
\tsc{BEC}
\tsc{DE}

\begin{document}
\let\WriteBookmarks\relax
\def\floatpagepagefraction{1}
\def\textpagefraction{.001}

\shorttitle{Epidemic Dynamics in Homes and Destinations under Recurrent Mobility Patterns}

\shortauthors{Yusheng Li et~al.}

\title [mode = title]{Epidemic Dynamics in Homes and Destinations under Recurrent Mobility Patterns}      

\affiliation[1]{organization={College of Artificial Intelligence},
    addressline={Southwest University}, 
    city={Chongqing},
    postcode={400715}, 
    country={PR China}}

\affiliation[2]{
    organization={Faculty of Natural Sciences and Mathematics, University of Maribor},
    addressline={Koroška cesta 160},
    city={Maribor},
    postcode={2000},
    country={Slovenia}
}

\affiliation[3]{
    organization={Community Healthcare Center Dr. Adolf Drolc Maribor},
    addressline={Ulica talcev 9},
    city={Maribor},
    postcode={2000},
    country={Slovenia}
}

\affiliation[4]{
    organization={Department of Physics, Kyung Hee University},
    addressline={26 Kyungheedae-ro},
    city={Seoul},
    postcode={02447},
    state={Dongdaemun-gu},
    country={Republic of Korea}
}

\affiliation[5]{
    organization={Complexity Science Hub},
    addressline={Metternichgasse 8},
    city={Vienna},
    postcode={1030},
    country={Austria}
}

\affiliation[6]{
    organization={University College, Korea University},
    addressline={145 Anam-ro},
    city={Seoul},
    postcode={02841},
    state={Seongbuk-gu},
    country={Republic of Korea}
}

\author[1]{Yusheng Li}[style=Chinese]

\author[1]{Yichao Yao}[style=Chinese]

\author[1]{Minyu Feng}[style=chinese,orcid=0000-0001-6772-3017]
\cormark[1]
\cortext[cor1]{Corresponding author}

\ead{myfeng@swu.edu.cn}

\author[2]{Tina P. Benko}[style=Chinese]

\author[2,3,4,5,6]{Matja{\v z} Perc}[style=Chinese]
\cormark[1]

\author[3]{Jernej Zavr{\v s}nik}[style=Chinese]


\begin{abstract}
    The structure of heterogeneous networks and human mobility patterns profoundly influence the spreading of endemic diseases. In small-scale communities, individuals engage in social interactions within confined environments, such as homes and workplaces, where daily routines facilitate virus transmission through predictable mobility pathways. Here, we introduce a metapopulation model grounded in a Microscopic Markov Chain Approach to simulate susceptible--infected--susceptible dynamics within structured populations. There are two primary types of nodes, homes and destinations, where individuals interact and transmit infections through recurrent mobility patterns. We derive analytical expressions for the epidemic threshold and validate our theoretical findings through comparative simulations on Watts--Strogatz and Barab\'asi--Albert networks. The experimental results reveal a nonlinear relationship between mobility probability and the epidemic threshold, indicating that further increases can inhibit disease transmission beyond a certain critical mobility level.
\end{abstract}

\begin{keywords}
  Epidemic modeling 
  \sep Human mobility 
  \sep Structured populations
  \sep Markov chain
  \sep Complex networks
\end{keywords}

\maketitle

\section{Introduction}
In recent years, the increasing importance of human mobility has played a critical role in the dynamic spreading of infectious diseases, particularly in small-scale communities where individuals follow routine movement patterns \cite{n1,n2,n3,r1,p1,li2022network}.
The distinctive structures of social networks and the recurrent mobility patterns of individuals are closely intertwined with the spatiotemporal dynamics of epidemics \cite{n4, n5, r2, r3,z0}.
Throughout history, epidemics such as the Black Death \cite{r4} were primarily propagated along specific trade routes or confined urban areas, which made the spatial and temporal spreading of diseases somewhat predictable. 
However, with the acceleration of economic globalization and expansion of transportation networks, epidemics nowadays exhibit characteristics of rapid dissemination and broader geographic reach \cite{z1,z2, z3}. The COVID-19 pandemic exemplified this phenomenon by simultaneously manifesting infection cases across multiple regions worldwide \cite{r5,n6}. Compared to previous centuries, the scale and speed of epidemic propagation have changed dramatically. The study of the spatial propagation of pathogens through the reaction-diffusion process has gained extensive application \cite{n7,n8,n9}. 
In the context of the complex networks, the reaction phase represents the transmission of pathogens within a subpopulation or a patch through direct contact, while the diffusion phase corresponds to the movement of infected hosts across different subpopulations via connected pathways or transportation links \cite{r6}.

The study of disease transmission dynamics within reaction-diffusion systems is commonly referred to as the metapopulation model \cite{r7,r8,r9}. The nodes of the metapopulation network represent population groups, while the links signify the migration of individuals between different patches. 
Initially introduced by Richard Levins \cite{r10} in the field of population biology, Anderson and May \cite{r11} were the first to apply the metapopulation concept to the SIR model in epidemiological studies, offering an effective framework for investigating spatial disease transmission. 
To gain a deeper understanding of the micro-scale transmission processes of diseases within local communities and the long-distance spreading of diseases due to human mobility, a series of metapopulation models based on the microscopic Markov chain approach have been employed \cite{r12,r13}. These models implement the dynamic process which involves movement-interaction-return (MIR) three stages. At the start of each time step, individuals either remain stationary or move to neighboring patches with a certain probability, followed by interactions and the spreading of the virus within each patch, and finally, return to their original locations to initiate the next cycle. 

While metapopulation models proved to be valuable in simulating recurrent human mobility patterns, many existing models assume that interactions are homogeneous mixing within communities, which oversimplifies the complex, heterogeneous nature of human contact patterns \cite{r14}. In reality, human social interactions are highly structured and unique, depending on both home and workplace environments.

In addition to models focusing on single network dynamics, recent research has explored structured population network models that account for the interplay between human mobility and social contacts \cite{r16,r17,r18}. 
Examples include heterogeneous social contact networks \cite{r19}, bipartite networks distinguishing between daytime and nighttime infections \cite{r20}, and models linking homes and public spaces in a bipartite structure \cite{r21}.
However, the conventional assumption that nodes with similar attributes exhibit identical statistical properties fails to capture the heterogeneity of human environments and the cyclical nature of commuting patterns, thus addressing the relationships between these factors and the evolving dynamics of pandemics as an open challenge \cite{r15}.

Building on these assumptions, in this work we utilize a Microscopic Markov Chain Approach (MMCA) to simulate the Susceptible-Infected-Susceptible (SIS) epidemic dynamics in structured heterogeneous populations characterized by recurrent mobility patterns. 
The network consists of two types of nodes: homes (e.g., hotels, dormitories), where individuals interact based on social ties, and destinations (e.g., workplaces, offices) where interactions are assumed to follow a well-mixed approximation. At each time step, individuals migrate between these two patch types following recurrent movement patterns, and infections occur via different interaction modes depending on the location.
The goal of this paper is to investigate how heterogeneous network structures and key factors such as mobility probability influence disease transmission. We provide a detailed description of the heterogeneous structured networks with recurrent mobility patterns and analytically derive the epidemic threshold, offering insights into the dynamics of epidemic spreading within small-scale communities.

The rest of the paper is presented as follows: In Section \ref{sec:model_description}, we introduce the formulation of our model by Markovian equations in constrained areas with recurrent mobility patterns. Subsequently, we describe the detailed infection dynamics at two different types of locations and derive the epidemic threshold. We conduct simulations and analyze the experimental results in Section \ref{sec:numerical_results}. Finally, we summarize the discussion and findings in Section \ref{sec:conclusion}.

\section{EPIDEMIC SPREADING IN STRUCTURED POPULATIONS UNDER RECURRENT MOBILITY PATTERNS}
\label{sec:model_description}

In our metapopulation model for epidemic spreading in structured populations, we consider \( N \) subpopulations, each comprising two types of locations: destinations and homes. Each subpopulation \( i \) has a population of \( n_i \) agents. In this framework, each node in the network represents a subpopulation (or patch) with varying population sizes. Destinations, such as workplaces or shopping centers, serve as interaction hubs where individuals engage in regular and frequent activities, and we model these interactions using a well-mixed approximation. In contrast to previous studies \cite{r20,r21} that assumed homogeneous subpopulations, we propose that interactions within homes are more constrained, occurring within personal social networks where individuals can only be infected by their immediate neighbors.

An important assumption in our model is that agents maintain their inherent social attributes, such as their degree of connectivity, when migrating to another subpopulation. Specifically, agents with a social connectivity of \( k \) (i.e., \( k \) neighbors) retain this same connectivity even after moving to a different community. This assumption helps explain the phenomenon of super-spreaders—individuals with higher-than-average contact rates due to their biological and behavioral traits, leading to a disproportionately high number of infections \cite{woolhouse1997heterogeneities}. Their elevated contact rates, driven by unique social behaviors, contribute significantly to the rapid spread of infections.

Building on this, each individual in home \( i \) interacts with approximately \( k_i \) neighbors, and these connections remain fixed, reflecting the stable nature of social ties in confined living environments. The movement of agents between adjacent subpopulations is governed by the weighted flow matrix \( W \), where \( W_{ij} \) indicates the connection strength between nodes \( i \) and \( j \). In practical scenarios, such as transportation networks, \( W_{ij} \) corresponds to the transportation throughput of routes connecting different cities.

\subsection{Model Description of Movement-Interaction-Return Patterns in Metapopulations} 
We construct the metapopulation model using the MMCA, which allows for a detailed representation of epidemic dynamics at the node level. This approach was initially introduced by Gómez-Gardenes in 2018 in the context of disease transmission dynamics, highlighting the influence of recurrent mobility patterns on reaction-diffusion processes in networks \cite{r12}.
The proposed model follows the process of Movement-Interaction-Return (MIR) patterns \cite{r13}, which models the dynamic stages of movement, interaction, and return, capturing the recurrent patterns of human commuting behavior more effectively. 

(1) Movement: At each time step, the migration pattern between different patches is determined by the mobility rate matrix $C$. 
The probability that an individual moves from their current location $i$ to patch $j$, with the mobility probability $p$, is proportional to the connection weight of elements $W_{ij}$ in the weighted flows matrix $W$, defined as
\begin{equation}
    C_{ij} = \frac{W_{ij}}{\sum_{j=1}^{N} W_{ij}}.
    \label{eq:Cij}
\end{equation}
Consequently, a fraction \( n_i p \) of agents from patch \( i \) will move to other patches, while the remaining \( n_i (1 - p) \) agents will stay at their current location.

(2) Interaction: Upon completing the movement process, agents engage in interactions within their new subpopulation.
Susceptible individuals have a probability \( \beta \) to become infected through contacting with an infected individual, while infected individuals recover at a rate \( \mu \), returning to the susceptible state. 

(3) Return: After the dynamical state update based on the epidemic model, each agent returns to its original patch and another reaction begins. 

Based on the assumptions above, we employ the classical SIS model to describe the transmission of the disease within the metapopulation network. Specifically, a susceptible individual has a probability $\beta$ of becoming infected upon contact with a contagious agent, while an infected individual will recover with a probability $\mu$ at the start of a time step and return to being susceptible.

Let \( \rho_i(t) \) denote the proportion of infected individuals in patch \( i \) at time step \( t \). The time evolution of \( \rho_i(t) \) is governed by
\begin{equation}
    \rho_i(t+1) = (1 - \mu)\rho_i(t) + (1 - \rho_i(t))\Pi_i(t).
    \label{eq:rho_i}
\end{equation}

Eq. (\ref{eq:rho_i}) can be interpreted as the proportion of infected individuals in patch $i$ at time $t + 1$. The first term on the right-hand side accounts for individuals who remain infected at time $t$ but have not recovered. The second term represents the newly infected individuals associated to patch $i$ during the current time step, determined by the infection probability \( \Pi_i(t) \), which is given by
\begin{equation}
    \Pi_{i}(t) = (1 - p) P_{i}(t) + p \sum_{j=1}^{N} C_{ij} P_{j}(t).
    \label{eq:Pi_i}
\end{equation}

To further refine $ P_i(t) $, the infection probability in (but not necessarily associated with) patch $ i $ at time $ t $, we consider the distinct roles of homes and destinations. \ref{fig:description} illustrates an example of the metapopulation network for epidemic propagation in our model.
On the metapopulation level (shown on the left-hand side), residential locations (homes), depicted by green circles, represent static social contact networks similar to Erd\H os-R\'enyi networks. In contrast, destination locations are assumed to be well-mixed and can be modeled as complete graphs. On the right-hand side of the figure, we show the movement between patch $ i $ and patch $ j $, governed by the probability matrix $ C $.

\begin{figure*}[ht]
    \centering
    \includegraphics [width=0.70\textwidth]{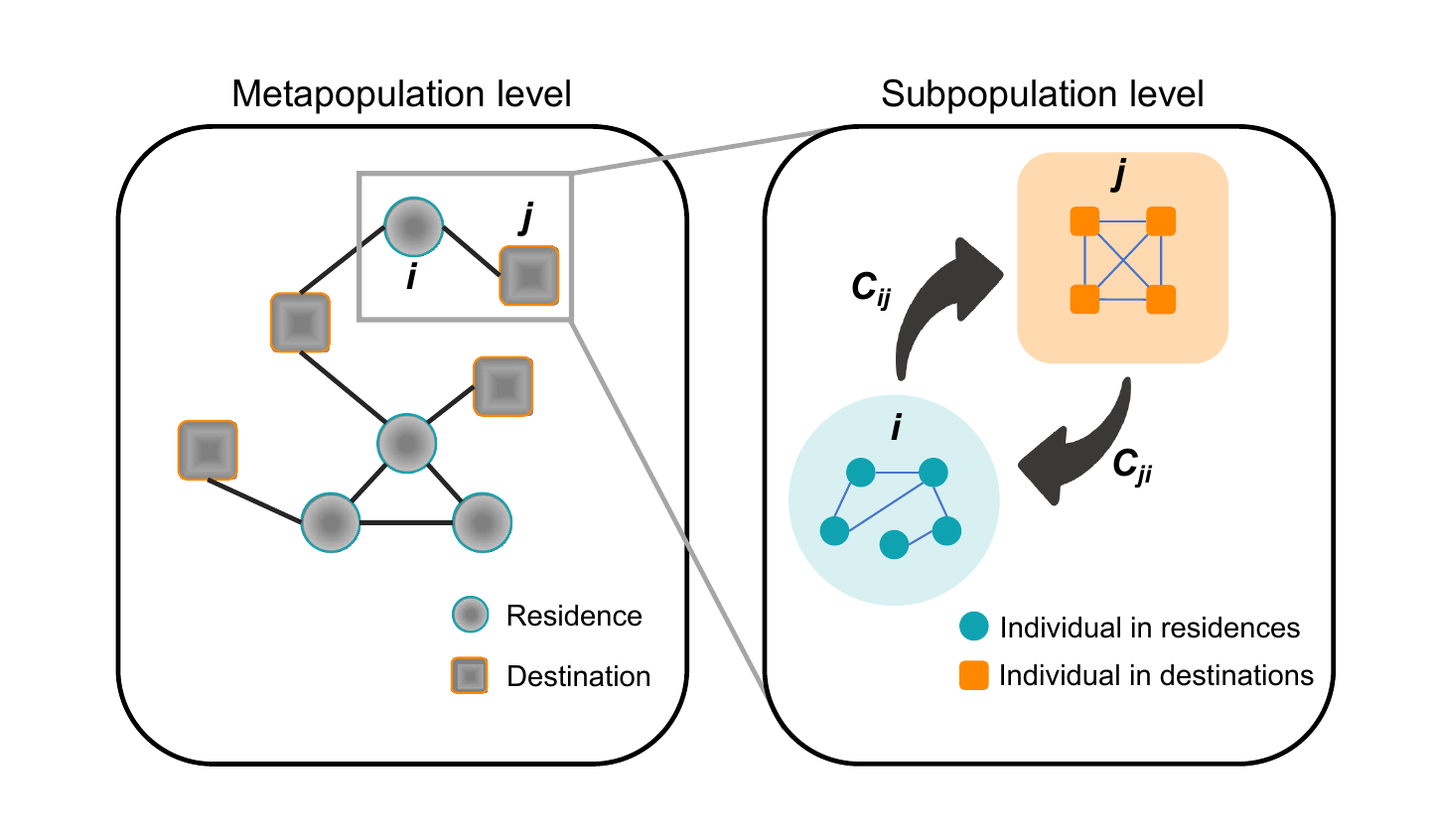}
    \caption{ 
        \textbf{Epidemic spreading in the heterogeneous metapopulation network with recurrent mobility patterns.} In the metapopulation network, each node represents a subpopulation or patch. Green circles represent residential locations (homes), while orange squares represent destination locations. In the subpopulation part of the illustration, solid squares indicate individuals interacting under a well-mixed approximation in destinations, and in homes, solid circles indicate epidemic spreading is driven by social contacts. The proportion of individuals in subpopulation $i$ moving to subpopulation $j$ depends on the weighted directed flow $C_{ij}$. Infection takes place in all the patches independently.
            }
    \label{fig:description}
\end{figure*}%

We denote the infection probabilities in homes and destinations as \( R_i(t) \) and \( D_i(t) \), respectively. In real-world scenarios, each patch can serve as both a home and a destination, with no clear distinction between them. To account for this, we introduce an activity coefficient \( \alpha \) with $0\le \alpha \le 1$ to denote the proportion of homes in the metapopulation network, and \( 1 - \alpha \) represents the proportion of destinations. Importantly, a lower \( \alpha \) value indicates higher activity, as it suggests that more individuals are at destinations (e.g., workplaces) where interactions are more frequent.

Thus, an individual will be infected at time $t$ in its own subpopulation $i$ with the infection probability:
\begin{equation}
    P_i(t) = \alpha R_i(t) + (1 - \alpha) D_i(t),
    \label{eq:P_i}
\end{equation}
where the first term represents the infection probability when patch \( i \) functions as a home, and the second term corresponds to the infection probability when patch \( i \) serves as a destination. This formulation allows us to capture the varying activity levels of subpopulations, with \( \alpha = 0 \) reflecting maximum activity in destinations and \( \alpha = 1 \) representing minimum activity in the limit condition.

\subsection{Infection Probabilities Regarding Two Types of Locations}
\label{sec:home_destination}
In general, the probability that a susceptible individual becomes infected after contact with $k$ contagious agents in a single network can be written as
\begin{equation}
    P(k) = 1 - (1 - \beta)^k,
    \label{eq:P_k}
\end{equation}
where $\beta$ denotes the infection probability.

Given the assumption that individuals at the destination 
$i$ are homogeneously mixed, such that each individual interacts with every other individual within the area, the infection probability for susceptible individuals at this destination can be expressed as 
\begin{equation}
    \begin{aligned}
        D_i(t) & = 1 - \prod_{j=1}^{N} (1 - \beta)^{n_{j \rightarrow i}\rho_j(t)} \\
        & = 1 - (1 - \beta)^{\sum_{j=1}^{N} n_{j \rightarrow i} \rho_j(t)},
        \label{eq:D_i_1}
    \end{aligned}
\end{equation}
where the exponent on the right hand represents the total number of infected individuals coming from all neighboring patches to destination $j$. Here, $n_{j \rightarrow i}$ is the number of individuals moving from node $j$ to destination $i$, denoted as
\begin{equation}
    n_{j \rightarrow i} = (1 - p)n_{i} \delta_{ij} + pC_{ji} n_{j},
    \label{eq:n_j_rightarrow_i}
\end{equation}
where $\delta_{ij} = 1$ when $i = j$ and $\delta_{ij} = 0$ otherwise.

For individuals in home $i$, we begin by calculating the average number of internal neighbors when $n_i p$ individuals leave. Given that each individual has an average of $k_i$ contact edges, the number of internal neighbors remaining when $n_i p$ individuals depart is
\begin{equation}
    N(i, \text{remained}) =  \frac{k_i}{n_i} n_i(1-p)=k_i(1-p),
    \label{eq:X_i_remained}
\end{equation}
where $n_i(1 - p)$ is the number of individuals staying in the home, and self-loops are not considered in the contact network. 

Assuming that the pathogen is uniformly distributed among the population, the total number of infected neighbors for each node becomes
\begin{equation}
    N(i,\text{infected})= (1 - p)k_i\rho_i(t) + p \sum_{j=1}^{N} C_{ji} k_j \rho_j(t),
    \label{eq:total_infected_neighbors1}
\end{equation}
where the first term accounts for infected neighbors remaining in the home, and the second term represents those coming from neighboring areas. 

Thus, the probability that a susceptible individual in home $i$ gets infected is
\begin{equation}
    R_i(t) = 1 - (1 - \beta)^{(1 - p)k_i\rho_i(t) + p \sum_{j=1}^{N} C_{ji} k_j \rho_j(t)}.
    \label{eq:R_i}
\end{equation}

To be more standardlized, we introduce $k_{j \rightarrow i}$ to denote the number of neighbors from subpopulation $j$ interacting with residents of $i$
\begin{equation}
    k_{j \rightarrow i} = (1 - p)k_i \delta_{ij} + p C_{ji} k_j.
    \label{eq:k_j_rightarrow_i}
\end{equation}

Combining Eqs. (\ref{eq:R_i}) and (\ref{eq:k_j_rightarrow_i}), the probability that a susceptible individual in home $i$ becomes infected is
\begin{equation}
    R_i(t) = 1 - (1 - \beta)^{\sum_{j=1}^{N} k_{j \rightarrow i} \rho_j(t)}.
    \label{eq:R_i_final}
\end{equation}

\subsection{Theoretical Derivation of the Epidemic Threshold}
\label{sec:epidemic_threshold}
In this section, we derive the epidemic threshold using Markov equations, which are essential for comprehending disease transmission dynamics in structured populations, particularly when considering human mobility between homes and workplaces. Within the metapopulation network framework with recurrent mobility patterns, our goal is to determine the critical infection rate that enables epidemic spreading. This derivation accounts for two key location types, homes and destinations, and captures the interactions between individuals within these locations, offering a more detailed understanding of the spreading. Finally, the largest eigenvalue of the matrix is employed to calculate the epidemic threshold for the SIS model, as presented in Theorem \ref{theorem1}.

\begin{theorem}
\label{theorem1}
For the classic SIS epidemic model, the epidemic threshold considering two types of locations in the metapopulation network with recurrent mobility is
$\beta_c=\frac{\mu}{\lambda_{max}(\mathbf{M})}$,
where $\lambda_{max}(\mathbf{M})$ is the largest eigenvalue of matrix $\mathbf{M}$, $\mathbf{M} = (\mathbf{m}_{ij})_{N\times N }=\alpha \mathbf{M}^a + (1-\alpha) \mathbf{M}^b$,
$\mathbf{m}^a_{ij} = (1 - p^2) \delta_{ij} k_j + p (1 - p) k_j (C + C^T)_{ij}+ p^2 k_j (C \cdot C^T)_{ij}$  and  $\mathbf{m}^b_{ij} =  (1 - p^2) \delta_{ij} n_j + p (1 - p) n_j (C + C^T)_{ij} 
+ p^2 n_j (C \cdot C^T)_{ij}$.
\end{theorem}

\begin{proof}
When epidemic spreading reaches a steady state ($t \rightarrow +\infty $), we can get the evolution with $\rho_i(t+1)=\rho_i(t)=\rho_i$. Under the assumption that near the critical onset of the epidemics, the fraction of infected individuals is negligible, we can substitute $\rho_i=\epsilon_i\ll 1$. Eq. (\ref{eq:rho_i}) then reads
\begin{equation}
    \epsilon_i=\epsilon_i(1-\mu)+(1-\epsilon_i)\Pi_i.
    \label{eq:epsilon_i}
\end{equation}

Substituting $\Pi_i$ and $P_i$ according to expressions in Eq. (\ref{eq:Pi_i}) and Eq. (\ref{eq:P_i}), we get
\begin{equation}
    \epsilon_i=\epsilon_i(1-\mu)+(1-\epsilon_i) [ (1-p)P_i+p\sum_{j=1}^{N}C_{ij}P_j  ] 
    \label{eq:epsilon_i_final},
\end{equation}
where
\begin{equation}
  \begin{aligned}
      P_i & = \alpha R_i+(1-\alpha )D_i                                            \\
          & = \alpha [1-(1-\beta)^{\sum_{j=1}^{N}k_{j\rightarrow i} \epsilon_j }]+
      \\&(1-\alpha )[1-(1-\beta)^{\sum_{j=1}^{N}n_{j\rightarrow i} \epsilon_j }]
      \label{eq:P_i_RD}.
  \end{aligned}
\end{equation}

Then we say that $\epsilon_i$ is small enough and apply the approximations $(1-\epsilon _i)^{n}\approx 1-n\epsilon _i$. Neglecting the second-order terms of $\epsilon_i$ and substituting $P_i$, $n_{j\rightarrow i}$ and $k_{j\rightarrow i}$ by their respective expressions in Eqs. (\ref{eq:P_i}), (\ref{eq:n_j_rightarrow_i}), (\ref{eq:k_j_rightarrow_i}), it follows that
\begin{equation}
  \begin{aligned}
      \Pi_i = & (1 - p) \Bigg[ \alpha  \beta \sum_{j=1}^{N} \epsilon_j \Big( (1 - p) k_j \delta_{ij} + p C_{ji} k_j  \Big)                  \\
              & + (1 - \alpha) \beta \sum_{j=1}^{N} \epsilon_j \Big( (1 - p) n_j \delta_{ij} + p C_{ji} n_j \Big) \Bigg]                    \\
              & + p \sum_{j=1}^{N} C_{ij} \Bigg[ \alpha  \beta \sum_{l=1}^{N} \epsilon_l \Big( (1 - p) k_l \delta_{jl} + p C_{lj} k_l \Big) \\
              & + (1 - \alpha) \beta \sum_{l=1}^{N} \epsilon_l \Big( (1 - p) n_l \delta_{jl} + p C_{lj} n_l \Big) \Bigg]
          \label{eq:Pi_i_threshold}.
  \end{aligned}
\end{equation}

Additionally, the following equation $\sum_{l=1}^{N}\epsilon_l k_l\delta _{jl}=\epsilon_jk_j$ has been used and substitute in Eq. (\ref{eq:epsilon_i}),
we further obtain
\begin{equation}
  \mu \epsilon_i= \beta\left [  \alpha (\mathbf{M}^a \vec{\epsilon})_i +(1-\alpha )(\mathbf{M}^b\vec{\epsilon})_i \right ]
  \label{eq:Pi_i_final},
\end{equation}
where the entries of matrix $\mathbf{M}^a$ read
\begin{equation}
  \begin{aligned}
    \mathbf{m}^a_{ij}  &= (1 - p^2) \delta_{ij} k_j 
    + p (1 - p) k_j (C + C^T)_{ij} + \\
    & p^2 k_j (C \cdot C^T)_{ij},
  \end{aligned}
  \label{eq:M1_ij}
\end{equation}
and the entries of matrix $\mathbf{M}^b$ are
\begin{equation}
  \begin{aligned}
    \mathbf{m}^b_{ij} & =  (1 - p^2) \delta_{ij} n_j + p (1 - p) n_j (C + C^T)_{ij} \\
    & + p^2 n_j (C \cdot C^T)_{ij}.
  \end{aligned}
  \label{eq:M2_ij}
\end{equation}

Denoting  $\mathbf{M} = \alpha \mathbf{M}^a + (1-\alpha) \mathbf{M}^b$, Eq. (\ref{eq:Pi_i_final}) can be rewritten as
\begin{equation}
    \frac{\mu}{\beta} \epsilon_{i}=\left(\mathbf{M} \vec{\epsilon}\right)_{i}.
\end{equation}
Thus, the epidemic threshold can be obtained by
\begin{equation}
    \beta_c=\frac{\mu}{\lambda_{max}(\mathbf{M})},
    \label{eq:beta_c}
\end{equation}
where $\lambda_{max}(M)$ is the largest eigenvalue of matrix $\mathbf{M}$. Whether the infection process occurs at a residential or destination location, each entry $\mathbf{m}_{ij}$ signifies the entire number of interactions between an individual at node \( i \) and all individuals connected with node \( j \). More specifically,  Eq. (\ref{eq:M1_ij})  denotes the overall average number of contacts within the same home, while Eq. (\ref{eq:M2_ij}) represents the average number of interactions at the destination.
\end{proof}

\begin{remark}
    In the case of $\alpha=0$, the infection process is confined to a single metapopulation network in which interactions within all subpopulations follow a well-mixed approximation.
    Therefore, we obtain the epidemic threshold $\beta_c=\frac{\mu}{\lambda_{max}(\mathbf{M})}$, where each element $\mathbf{m}_{ij} $ of matrix $\mathbf{M}$ can be replaced by $ \mathbf{m}_{ij} =  (1 - p^2) \delta_{ij} n_j + p (1 - p) n_j (C + C^T)_{ij} 
    + p^2 n_j (C \cdot C^T)_{ij}.$

\end{remark}

\section{NUMERICAL RESULTS}
\label{sec:numerical_results}
To systematically validate the effectiveness of our model in small-scale communities, we conducted extensive simulation experiments on metapopulation networks consisting of \( N = 50 \) subpopulations. The number of agents within each patch follows a random distribution between 50 and 150, resulting in an average total of 5000 agents across the entire metapopulation network after over 100 repeated simulations.

In our experiments, we consider two distinct metapopulation networks: the Watts-Strogatz Small-World (WS) network and the Barabási-Albert (BA) network. Given the significant influence of contact network topology on disease transmission within the patches \cite{r8, r11}, we assume a well-mixed network structure for destinations, where individuals interact more frequently. For patches referred to as homes, we employ Erdős-Rényi (ER) networks as the social contact networks, which exhibit a degree distribution following a Poisson distribution when the mobility probability \( p \) is small. To ensure consistency between the two network structures, we maintain the average degree of two different types of network structures nearly identical (e.g., \( \langle k \rangle = 10 \)).
Moreover, a key factor is the migration matrix \( C \), which governs the probability of individuals moving between subpopulations. The elements \( C_{ij} \) represent the probability that an individual in subpopulation \( i \) will move to a neighboring subpopulation \( j \), with values ranging from 0 to 1. The weight matrix \( W \) is designed to reflect the connectivity between communities, with higher values indicating stronger migration links between subpopulations.
The numerical results of the proposed model are derived using Markov equations and evaluated through Monte Carlo (MC) simulations.

The experiments for the proposed model begin with 1\% of infected individuals in each patch, which is also applied in the MC simulations unless otherwise specified. When the system reaches equilibrium after 500 time steps, we compute the average results over 100 simulations to ensure the accuracy of the experiments and eliminate the influence of randomness.

\subsection{Comparision between theoretical results and Monte Carlo simulations}

\begin{figure*}[ht!]
    \centering
    \begin{subfigure}[b]{0.45\textwidth}
        \centering
        \includegraphics[width=\textwidth]{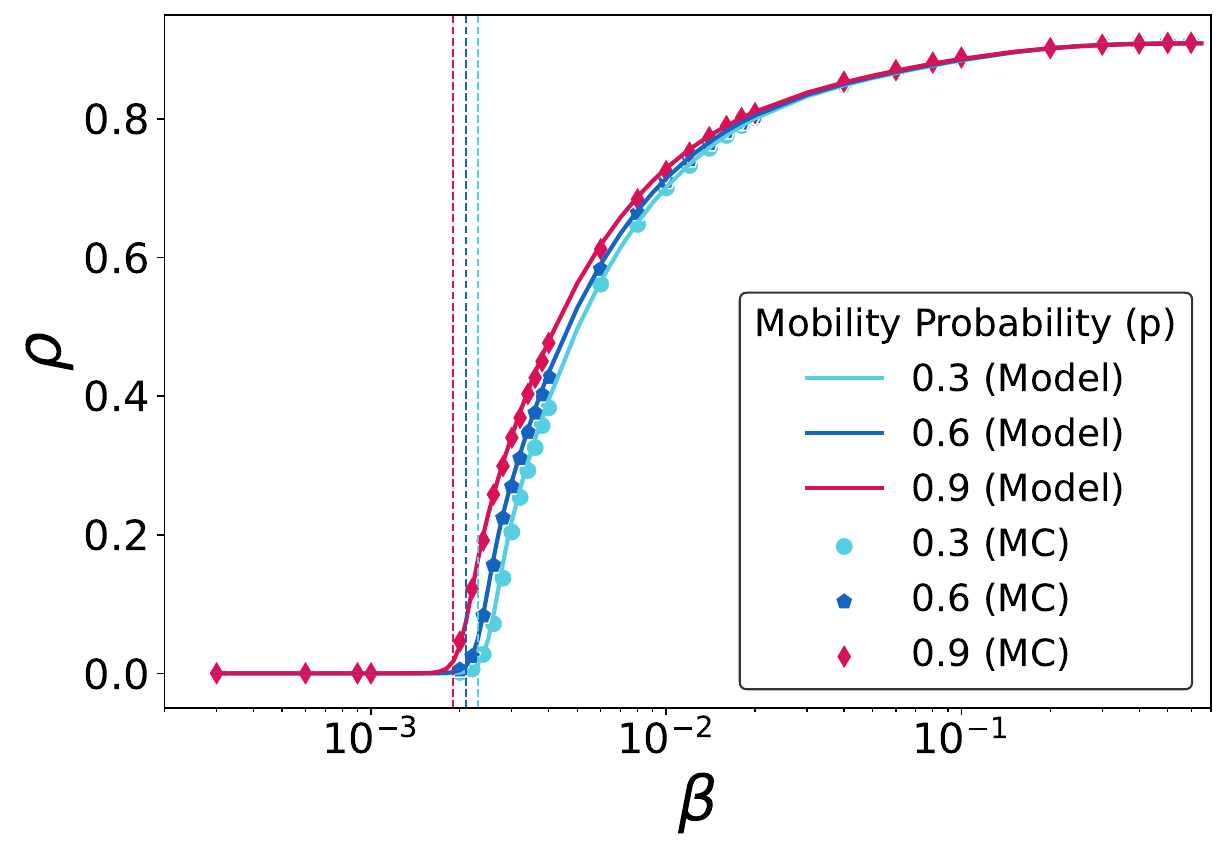}
        \caption{}
        \label{fig_final_1a}
    \end{subfigure}
    \hspace{0.02\textwidth}
    \begin{subfigure}[b]{0.45\textwidth}
        \centering
        \includegraphics[width=\textwidth]{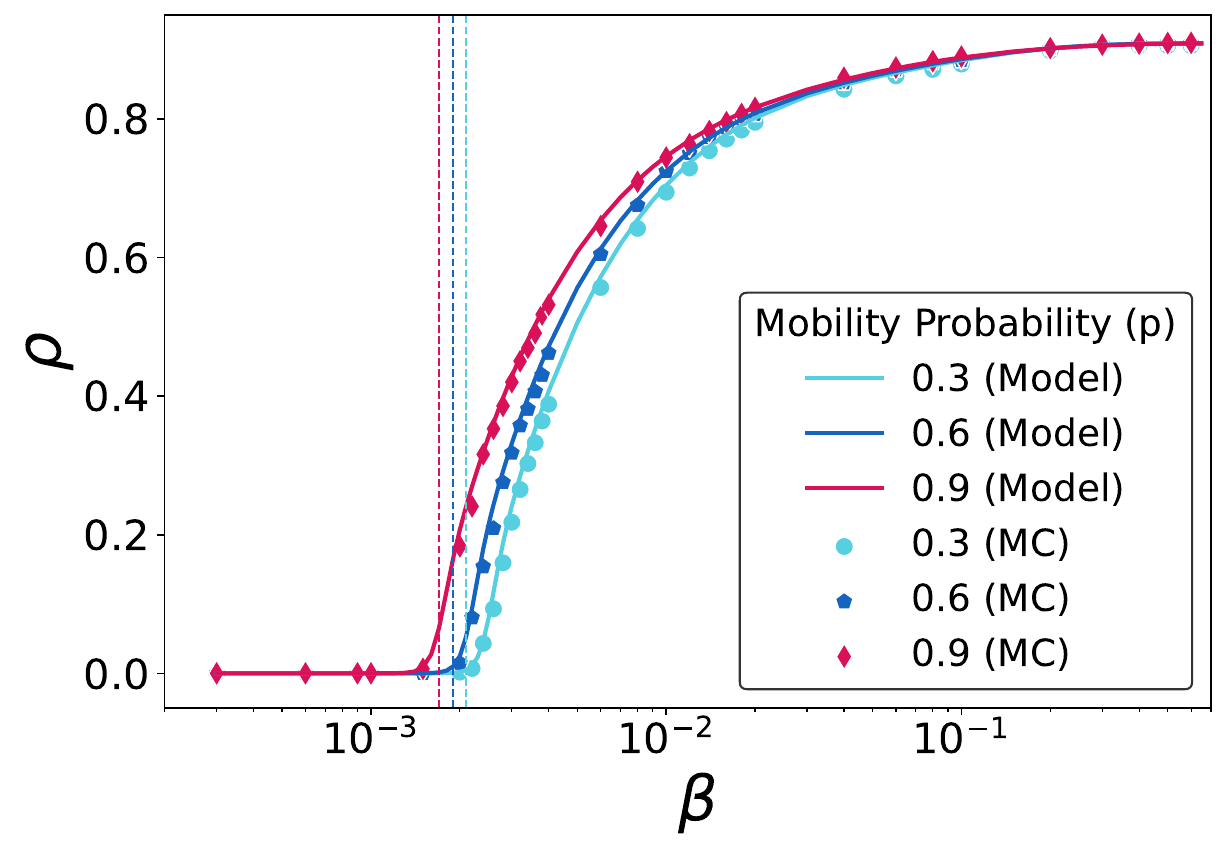}
        \caption{}
        \label{fig_final_1b}
    \end{subfigure}
    \caption{
        \textbf{Comparison of the theoretical results with the Monte Carlo simulations for Watts-Strogatz network and Barabási-Albert network.} The proportion of infected individuals $\rho$ in the steady state as a function of infection rate $\beta$ for three values of the mobility probabilities in two types of networks. Solid lines represent the results of the proposed model, dashed vertical lines denote the epidemic threshold obtained through Eq. (\ref{eq:beta_c}) and the symbols are average results calculated by Monte Carlo Method through 50 simulations. For each graph,  there are a total of 5000 individuals in 50 locations. The activity coefficient is set to \( \alpha = 0.6 \). The recovery rate is set to \( \mu = 0.1 \). (a) Watts-Strogatz network. (b) Barabási-Albert network.
    }
    \label{fig_1}
\end{figure*}

Firstly, we compare the theoretical results obtained by the MMCA method with the MC simulations, as shown in Fig. \ref{fig_1}.
We show the proportion of infected individuals $\rho$ in the steady state of the entire metapopulation network as a function of infection density $\beta$, for three distinct mobility probabilities \( p = 0.3 \), \( 0.6 \), and \( 0.9 \). Theoretical curves are represented by solid lines, while different markers indicate the MC simulation results, with epidemic thresholds marked by dashed vertical lines based on Eq. (\ref{eq:beta_c}). This setup includes 30 residential nodes and 20 destination nodes within the metapopulation, governed by an activity coefficient \( \alpha = 0.6 \).

The results indicate that for lower infection rates \( \beta \), a slight discrepancy emerges between the analytical predictions and simulation results, with the MC outcomes for both WS and BA networks being marginally lower than the theoretical expectations.
Additionally, the BA network exhibits a more pronounced difference in epidemic thresholds for different mobility probabilities, evidenced by the larger gaps between curves in Fig. \ref{fig_final_1b}. 
For sufficiently large infection rates \( \beta \), the theoretical results are in good correspondence with the Monte Carlo simulations in the steady state.
As the mobility probability $p$ increases, the epidemic threshold derived from theoretical analysis progressively decreases relative to the threshold obtained from Eq. (\ref{eq:beta_c}) for both network structures.
Furthermore, comparing Fig. \ref{fig_final_1a} and Fig. \ref{fig_final_1b}, it is evident that for the same mobility probability and infection rate, the epidemic threshold is lower in the BA network than in the WS network. This can be attributed to the heterogeneity in the degree distribution of BA network, which results in a higher maximum eigenvalue, thereby increasing the likelihood of an outbreak in the power-law network structure.

\subsection{Effect of mobility probability and infection rate on final infection density}

\begin{figure*}[ht!]
    \centering
    \begin{subfigure}[b]{0.45\textwidth}
        \centering
        \includegraphics[width=\textwidth]{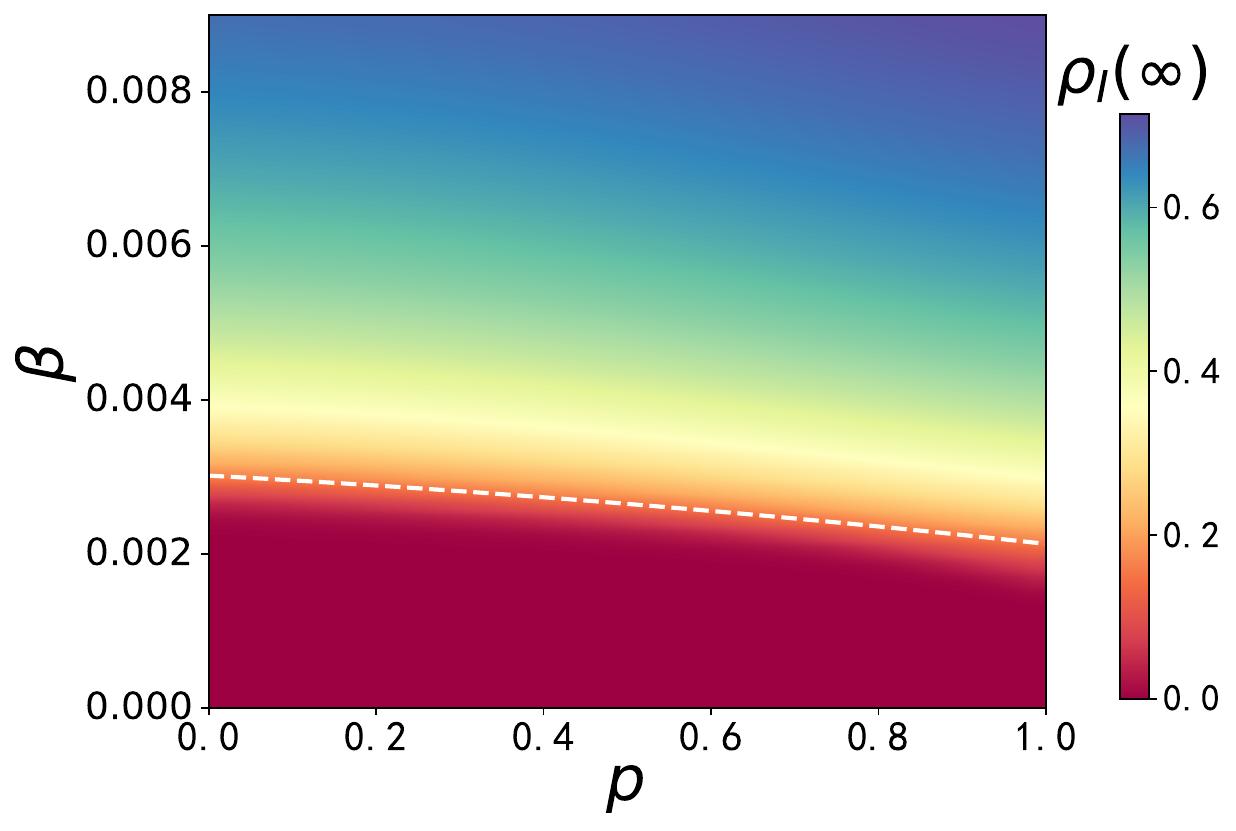}
        \caption{}
        \label{fig_final_2a}
    \end{subfigure}
    \hspace{0.02\textwidth}
    \begin{subfigure}[b]{0.45\textwidth}
        \centering
        \includegraphics[width=\textwidth]{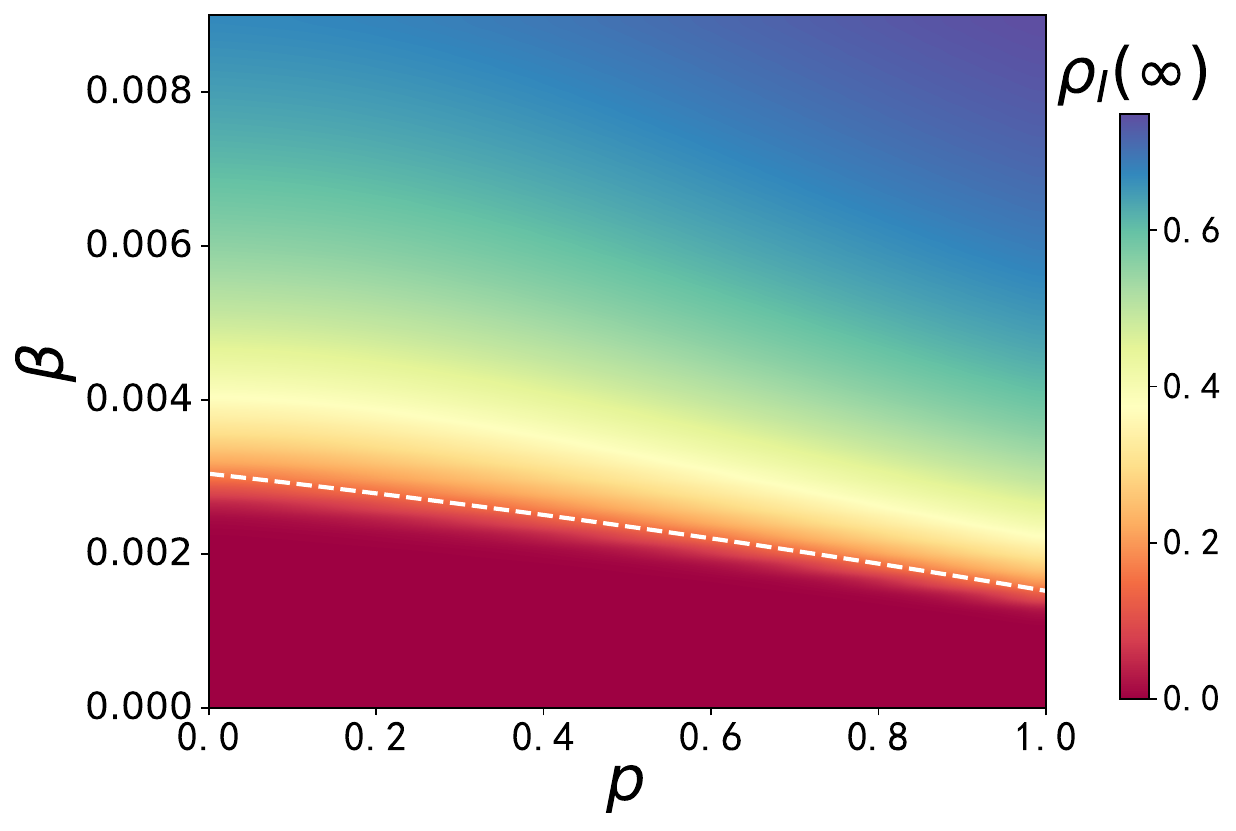}
        \caption{}
        \label{fig_final_2b}
    \end{subfigure}
    \caption{
        \textbf{Final infection density changes with sets of mobility probability and infection rate under different structures of networks.}  The heatmaps show the final infection density obtained through theoretical analysis, while the white dashed lines represent the epidemic thresholds from Monte Carlo simulations. The experiment consists of 5000 individuals distributed across 50 patches, with the recovery rate \( \mu = 0.1 \) and the active coefficient \( \alpha = 0.6 \). (a) Watts-Strogatz network. (b) Barabási-Albert network.}
    \label{fig_2}
\end{figure*}

To further explore the impact of mobility probability \( p \) and infection rate \( \beta \) on the final infection density $\rho_I \left ( \infty \right )$ and epidemic threshold, we accomplish a series of experiments by setting a range of values for $\beta$ and $p$ in Watts-Strogatz and Barabási-Albert metapopulation networks, respectively. As reported in Fig. \ref{fig_2},the results indicate that irrespective of the mobility probability \( p \), the disease will propagate and reach a steady state when the infection rate \( \beta \) exceeds a critical value.
Notably, increased mobility probability accelerates disease spreading, resulting in a lower epidemic threshold. Comparing the white dashed lines between Figs. \ref{fig_final_2a} and \ref{fig_final_2b}, which display the epidemic thresholds obtained from MC simulations, we observe that the epidemic threshold decreases more rapidly in the BA network than that in the WS network. It is the same conclusion as we get in Fig. \ref{fig_1}.
However, this difference becomes negligible as the experimental scale increases in our expanded simulations, with more connections, patches, and a larger population within each patch.

\subsection{Impact of mobility and activity coefficient on epidemic threshold}

\begin{figure*}[ht!]
    \centering
    \begin{subfigure}[b]{0.45\textwidth}
        \centering
        \includegraphics[width=\textwidth]{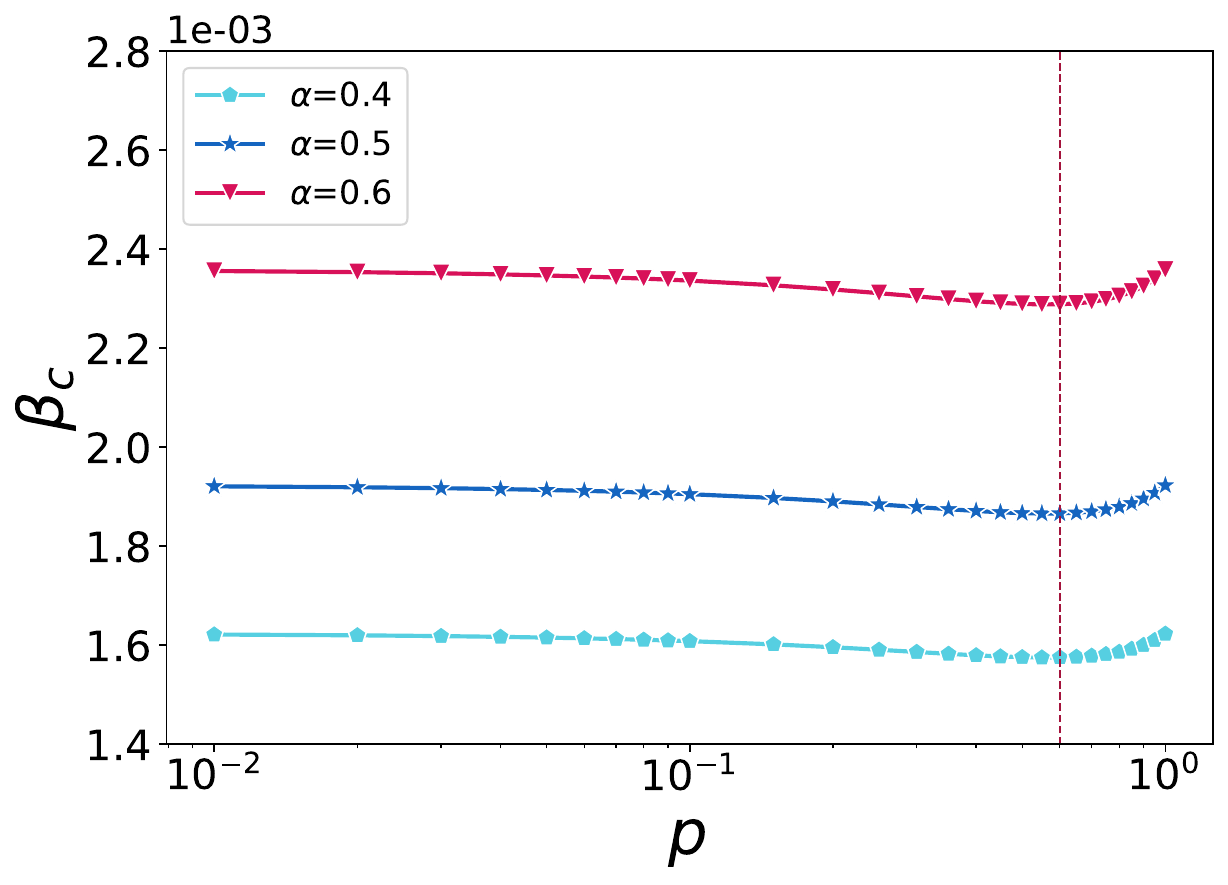}
        \caption{}
        \label{fig_final_3a}
    \end{subfigure}
    \hspace{0.02\textwidth}
    \begin{subfigure}[b]{0.45\textwidth}
        \centering
        \includegraphics[width=\textwidth]{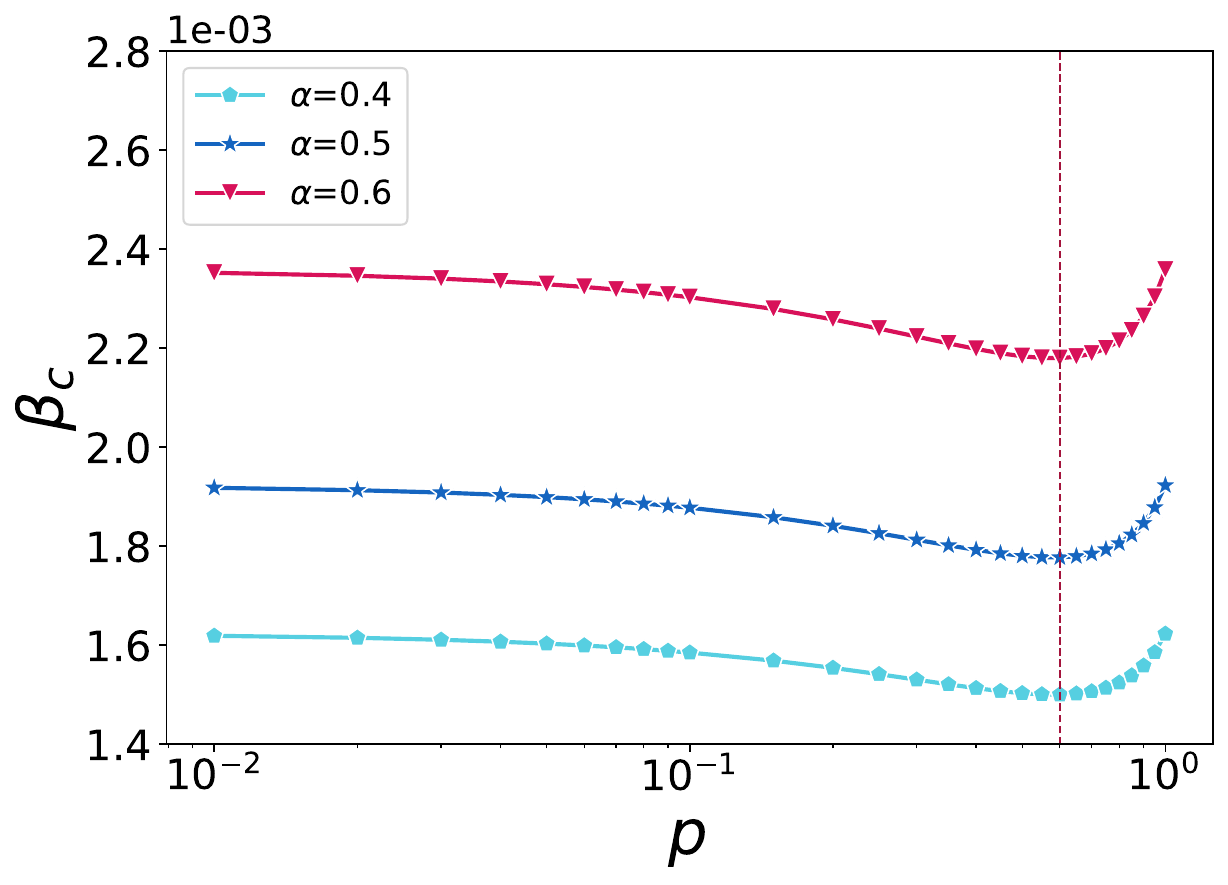}
        \caption{}
        \label{fig_final_3b}
    \end{subfigure}
    \caption{
        \textbf{Comparison of epidemic thresholds \( \beta_c \) with mobility probabilities \( p \) for different values of \( \alpha \).} Each curve shows the trend of \( \beta_c \) as a function of \( p \) across varying values of \( \alpha \), indicating different ratios of home to destination sites. The red dashed lines mark the critical value of \( p \) at which the epidemic threshold is minimized, as derived from Eq. (\ref{eq:beta_c}). The simulation includes 5000 individuals distributed across 50 patches, with a recovery rate \( \mu = 0.1 \). (a) Watts-Strogatz network. (b) Barabási-Albert network.}
    \label{fig_3}
\end{figure*}

\begin{figure*}[ht!]
    \centering
    \begin{subfigure}[b]{0.45\textwidth}
        \centering
        \includegraphics[width=\textwidth]{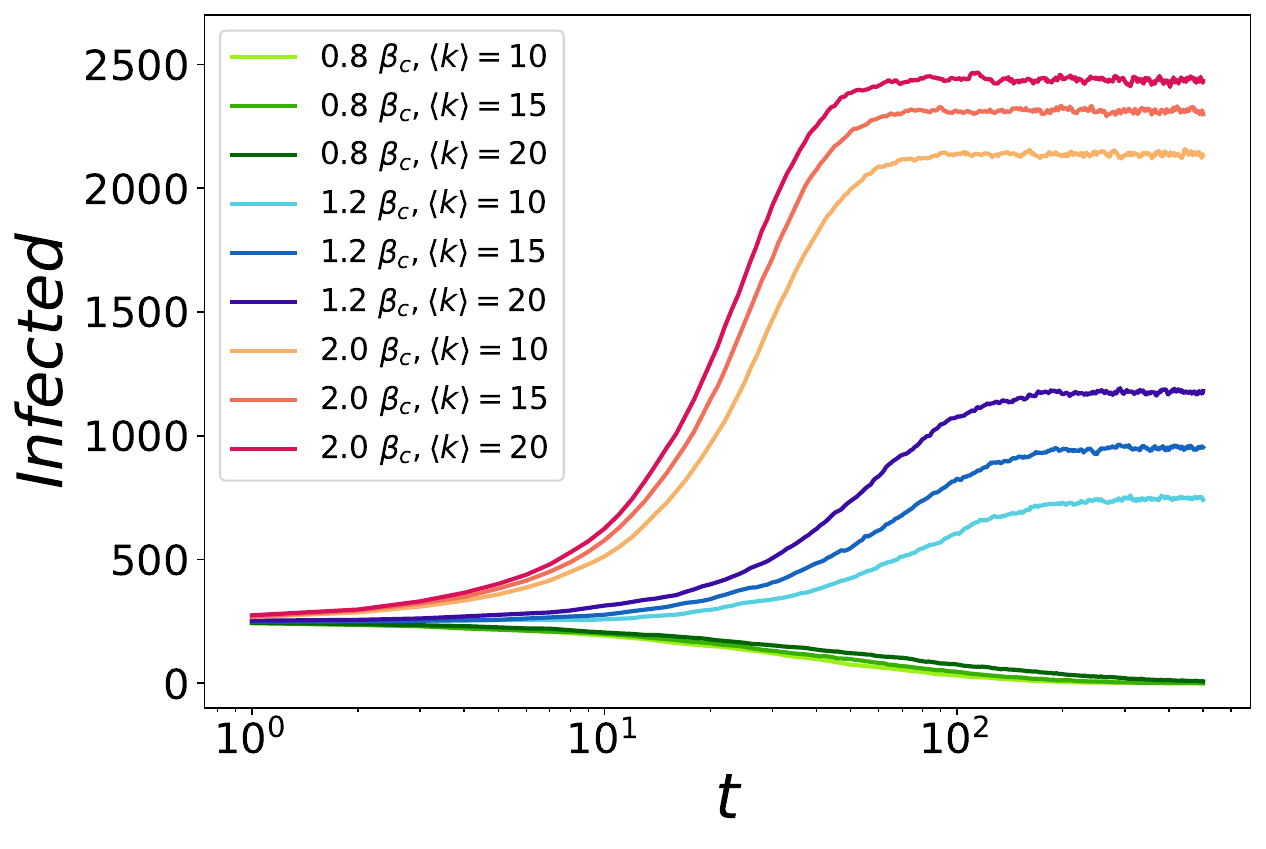}
        \caption{}
        \label{fig_final_4a}
    \end{subfigure}
    \hspace{0.02\textwidth}
    \begin{subfigure}[b]{0.45\textwidth}
        \centering
        \includegraphics[width=\textwidth]{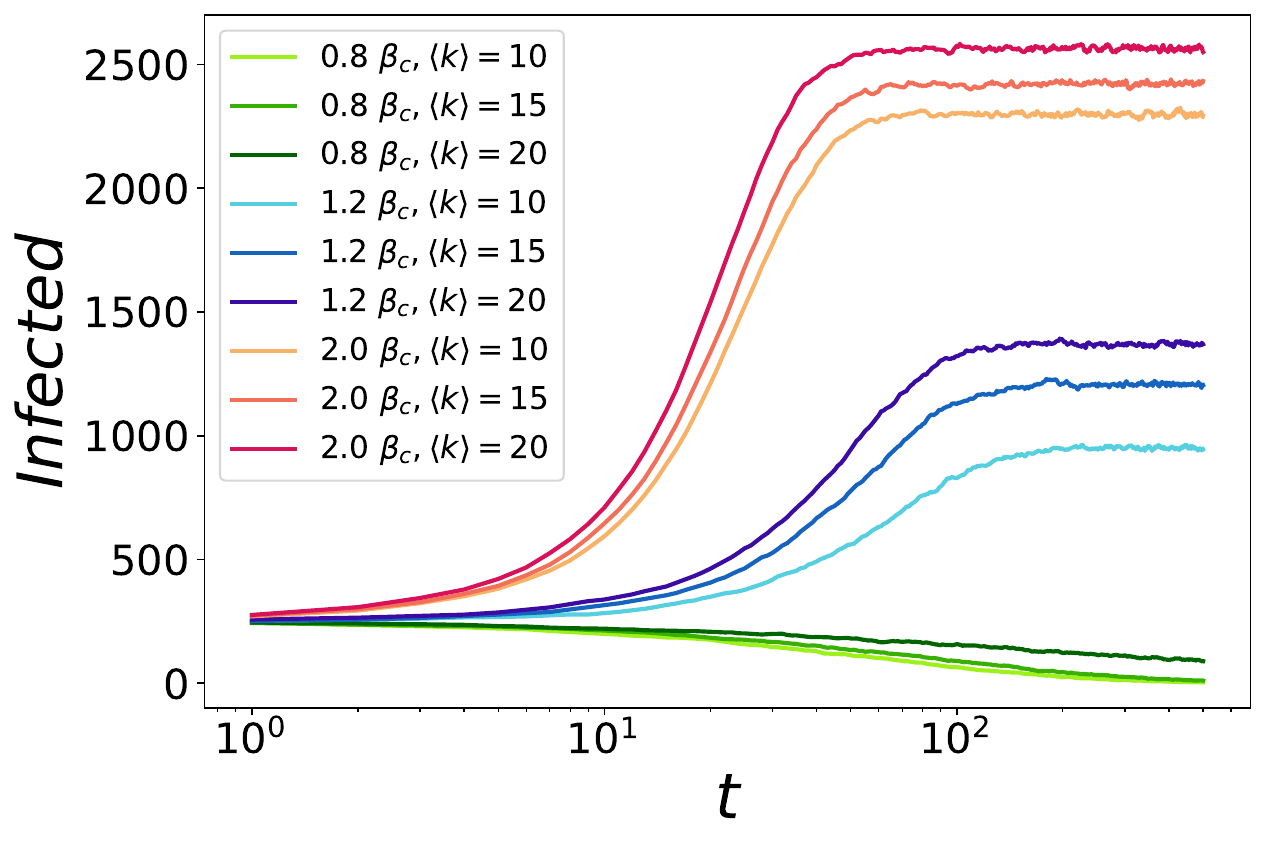}
        \caption{}
        \label{fig_final_4b}
    \end{subfigure}
    \caption{
        \textbf{Evolution of the number of infected individuals for different infection rates $\beta$ and average contact degrees $ \langle k \rangle$ varying over time.} There are 5000 individuals in 50 patches, with recovery rate $\mu$ = 0.1, $\alpha$ = 0.6, and an initial number of 250 infected individuals. The mobility probability is set to \( p = 0.6 \), and the simulations are run for 500 time steps.  Each curve is averaged over 20 MC simulation runs for each combination of \( \beta \) and \( \langle k \rangle \). (a) Watts-Strogatz network. (b) Barabási-Albert network.}
    \label{fig_4}
\end{figure*}

In this subsection, we explore the impact of mobility probability \( p \) and activity coefficient \( \alpha \) on the epidemic threshold \( \beta_c \), examining how variations in different types of locations influence disease transmission dynamics. Individuals are assumed to engage in lower social activity within residential areas but become more active in destinations where daily interactions occur. By adjusting the proportion of time spent at homes versus destinations, represented by the parameter \( \alpha \), we aim to capture the effects of heterogeneous population structures on disease spreading. Fig. \ref{fig_3} presents the relationship between epidemic threshold \( \beta_c \), mobility probability \( p \), and activity coefficient \( \alpha \) in both WS and BA networks. Each curve in the figure represents the epidemic threshold trend for different values of \( \alpha \) (0.4, 0.5, and 0.6), reflecting three scenarios with varying proportions of homes and destinations.

From Fig. \ref{fig_3}, we observe that a lower \( \alpha \) value corresponds to a reduced epidemic threshold \( \beta_c \), with the minimum threshold occurring when the number of destination sites exceeds that of homes (\( R < D \)), represented by \( \alpha = 0.4 \). Conversely, the highest thresholds are observed when \( \alpha = 0.6 \) (\( R > D \)). For small mobility probabilities \( p \) between 0.01 and 0.1, \( \beta_c \) decreases almost linearly in both WS and BA networks. However, as \( p \) exceeds 0.1, we identify a critical value of \( p \), indicated by the red dashed line, where \( \beta_c \) reaches its minimum and subsequently increases as mobility continues to rise. This non-monotonic behavior indicates a counterintuitive finding that higher mobility, beyond a certain threshold, can reduce the risk of epidemic spreading. Similar phenomena are also concluded in \cite{r12, r16}. Furthermore, epidemic thresholds in the BA network are generally lower than those in the WS network.

\subsection{Effect of infection rate and average contact degree on epidemic spreading}
In the last simulation, we investigate the impact of the degree distribution of social contact networks in homes on epidemic spreading. Fig. \ref{fig_4} illustrates the evolution of the total number of infected individuals over time for various infection rates \( \beta \) and average contact degrees \( \langle k \rangle \). 
Regardless of the average degree of the contact network, the epidemic dies out exponentially when \( \beta < \beta_c \), but spreads through the metapopulation network when \( \beta \geq \beta_c \). We can also conclude that, for the same infection rate, such as $\beta_c =1.2$ and $\beta_c=2.0$, the infection increases as the average contact degree increases. Both the speed of disease transmission and the total number of infected individuals are positively correlated with \( \langle k \rangle \), as shown in Fig. \ref{fig_4}.
This relationship can be attributed to the increased contact frequency, which reduces the heterogeneity of social interactions and makes the disease transmission more similar to homogeneous mixing.
Additionally, for a fixed threshold, the number of infected individuals in the steady state is higher in the BA networks than in the WS networks. This result is consistent with previous simulations, where the proposed model demonstrated that epidemics spreading more rapidly and extensively in scale-free networks like the Barabási-Albert network, leading to a larger overall scale of infection.

\section{CONCLUSION}
\label{sec:conclusion}
In this paper, we introduce a heterogeneous metapopulation model that incorporates recurrent mobility patterns within confined areas, capturing the dual roles of homes and destinations in disease transmission. homes are represented as structured social contact networks where individuals interact locally, while destinations are modeled with a well-mixed approximation to account for more active social interactions.
Using the MMCA, we analyze epidemic dynamics within this framework and derive the epidemic threshold in the steady state. We conduct extensive simulations on WS and BA networks, comparing theoretical predictions with MC simulations. Our findings reveal a strong alignment between theoretical and simulation results, with the BA network exhibiting a lower epidemic threshold and faster disease spreading than the WS network under similar conditions. Additionally, we identify a non-monotonic relationship between mobility probability and the epidemic threshold, indicating that mobility may exacerbate epidemic spreading beyond a critical value. Our analysis also demonstrates that when the infection rate is below the epidemic threshold, the disease consistently dies out, irrespective of network topology. In summary, the proposed metapopulation model enhances our understanding of disease spreading driven by human mobility in restricted environments. Its flexibility in accommodating heterogeneous networks with varied population sizes, weighted connections, and diverse structural configurations makes it a valuable tool for investigating epidemic processes in real-world settings and offers practical insights for designing interventions in public health and epidemic control.

Nonetheless, there remain limitations that require further exploration. Our model assumes a constant mobility rate for individuals and does not account for public awareness or behavioral adjustments, which could be influenced by factors such as age, gender, and geographic location. Additionally, analyzing social connections and interaction patterns within each subpopulation in greater detail poses considerable challenges. Future studies will aim to address these limitations by incorporating more complex mobility behaviors and social dynamics into the model.

\section*{ACKNOWLEDGMENTS}
This work was supported by the National Natural Science Foundation of China (NSFC) (Grant No. 62206230),  
the Natural Science Foundation of Chongqing (Grant No. CSTB2023NSCQ-MSX0064),
and the Slovenian Research and Innovation Agency
(Javna agencija za znanstvenoraziskovalno in inovacijsko dejavnost Republike Slovenije) (Grant Nos. P1-0403 and N1-0232).

\bibliographystyle{model1-num-names}

\bibliography{cas-refs}


\end{document}